\documentclass{llncs}
\usepackage{amsfonts}
\usepackage{amsmath}
\usepackage{url}
\usepackage{bm}
\usepackage{placeins}
\usepackage{graphicx,xspace}
\usepackage{caption}
\usepackage{subcaption}
\usepackage{mdframed}
\usepackage{tikz}

\let\doendproof\endproof{}
\renewcommand\endproof{~\hfill$\qed$\doendproof}

\newif\ifabstract{}
\abstracttrue{}
\newif\iffull{}
\ifabstract{} \fullfalse{} \else \fulltrue{} \fi

\newtoks\magicAppendix{}
\magicAppendix={}
\newtoks\magictoks{}
\newif\iflater{}
\laterfalse{}
\ifabstract{}
\long\def\later#1{\magictoks={#1}%
\edef\magictodo{\noexpand\magicAppendix={\the\magicAppendix{} \par
    \noexpand\setcounter{proposition}{\arabic{proposition}}%
    \noexpand\setcounter{theorem}{\arabic{theorem}}%
    \noexpand\setcounter{lemma}{\arabic{lemma}}%
    \noexpand\setcounter{definition}{\arabic{definition}}%
    \noexpand\setcounter{corollary}{\arabic{corollary}}%
    \noexpand\setcounter{remark}{\arabic{remark}}%
    \the\magictoks}}%
  \magictodo}
\long\def\both#1{\later{#1}\the\magictoks}
\else
\long\def\both#1{#1}
\long\def\later#1{#1}
\fi
\def\magicappendix{\latertrue{} \the\magicAppendix}

\newcommand{\ep}{\varepsilon}
\newcommand{\LL}{\mathcal{L}}
\newcommand{\mapmrf}{\textsc{Pairwise MAP MRF}\xspace}

\DeclareMathOperator{\OPT}{OPT}
\DeclareMathOperator{\poly}{poly}

\DeclareMathOperator*{\argmax}{arg\,max}

\title{PTAS for MAP Assignment on
       Pairwise Markov Random Fields in Planar Graphs}

\author{%
Eli Fox-Epstein
\and
Roie Levin
\and
David Meierfrankenfeld
}

\institute{%
  Department of Computer Science,
  Brown University,
  Providence RI\@.
  \texttt{\{ef,rl46,nfelddav\}@cs.brown.edu}
}

\begin{document}

\maketitle

\begin{abstract}
We present a PTAS for computing the maximum a posteriori assignment
  on Pairwise Markov Random Fields with
  non-negative weights in planar graphs.
This algorithm is practical and not far behind
  state-of-the-art techniques in image processing.
MAP on Pairwise Markov Random Fields with (possibly) negative weights
  cannot be approximated unless P = NP, even on planar graphs.
We also show via reduction that this yields a PTAS
  for one scoring function
  of Correlation Clustering in planar graphs.
\end{abstract}

\section{Introduction}

Pairwise Markov Random Fields (MRFs) model distributions in a variety of applications 
  and arise in fields as diverse as statistical physics, computer vision, coding theory,
  computational biology, machine learning, and combinatorial optimization.
Solving associated optimization problems is critical in practice and
  also of high theoretical importance.

We briefly review the statistical view on MRFs before focusing on the combinatorial
  problem.
A pairwise MRF is a set of $n$ random variables $\bm{X} = \{X_1, \ldots, X_n\}$ over
  label set $\{1, \ldots, L\}$,
  a graph~$G = (\bm{X}, E)$, where
\[
\Pr[\bm{X}=\bm{x}] = \frac{1}{Z}\exp{\left(\sum\limits_{i \in V} \phi_i(x_i) + \sum\limits_{(i,j) \in E} \psi_{ij}(x_i,x_j)\right)},
\]
where $\phi_i$ and $\psi_{ij}$ are arbitrary functions
  and $Z$ is a normalizing constant.
Intuitively, $\phi_i(x_i)$ can be regarded as vertex $i$'s preference
  for label $x_i$ and $\psi_{ij}(x_i, x_j)$ as the compatibility
  between labels $x_i$ and $x_j$ on the endpoints of edge $ij$.
We are interested in finding a maximum a posteriori (MAP) assignment~$\bm{x}^*$,
  i.e.~$\bm{x}^* = \argmax_{\bm{x}} \Pr[\bm{X} = \bm{x}]$.
Finding the MAP label assignment corresponds to this optimization problem:
\begin{mdframed}[nobreak=true]
  \noindent \mapmrf{}

  \begin{description}
    \item[Instance:]\ 
      \begin{itemize}
        \item graph $G = (V, E)$,
        \item label set $\LL = \{1, \ldots, L\}$,
        \item singleton functions $\phi_i(\cdot) : \LL \to \mathbb{R} \quad \forall \; i \in V$,
        \item pairwise functions $\psi_{ij}(\cdot, \cdot) : \LL \times \LL \to \mathbb{R} \quad \forall \; (i, j) \in E$.
      \end{itemize}

    \item[Solution:] for each $v \in V$, label assignment $x_v \in \LL$

    \item[Maximize:]
      \[
      H(\bm{x}) = \sum\limits_{v \in V} \phi_i(x_v) + \sum\limits_{(u,v) \in E} \psi_{uv}(x_u,x_v),
      \]
  \end{description}
\end{mdframed}
Throughout the paper, we will assume $G$ to be connected;
  combining the solutions on each component handles the case of
  a disconnected graph.

\mapmrf{} has been considered in many domains, however, in general graphs we will show:

\both{%
\begin{theorem}\label{thm:general-hard-to-approximate}
  There is an $\alpha > 0$ such that, unless P = NP,
  there is no polynomial-time $\alpha$-approximation algorithm for \mapmrf{},
    even for nonnegative $\phi$ and $\psi$.
\end{theorem}
}

\later{%
\begin{proof}
\textsc{Maximum Cut} is NP-hard to approximate to better than a~$60/61$
  factor~\cite{Trevisan}.
There is an approximation-preserving reduction from \textsc{Max Cut} to \mapmrf{},
  by setting~$\phi_i(x_i) = 0$ and~$\psi_{ij}(x_i, x_j)$ to be 1
  if~$x_i \neq x_j$ and 0 otherwise.
\end{proof}
}

In light of this, we focus on planar graphs
  as many real-world instances, such as those from computer vision,
  are planar or nearly planar.
It turns out that \mapmrf{} is still NP-hard on planar graphs~\cite{Barahona}.
However, restricting our attention to planar graphs allows for much better
  approximation algorithms.

We will additionally require $\phi_i$ and $\psi_{ij}$ to be nonnegative.
By setting
\begin{align*}
  \phi_i'(x) &= \phi_i(x)-\min_{a\in \LL}(\phi_i(a)) && \forall \; i\in V, x\in \LL && \text{and} \\
  \psi_{ij}'(x,y) &= \psi_{ij}(x,y)-\min_{a,b\in \LL}(\psi_{ij}(a,b)) && \forall \; (i,j) \in E,\; \forall \; x,y\in \LL,
\end{align*}
we can transform an instance with general weights into an instance
with non-negative weights with the same optimal assignment.
However, this changes the value of the objective function,
  and thus also the approximation ratio.
This restriction is necessary, as with general weights
  \mapmrf{} is impossible to approximate unless P = NP\@.
In particular:

\both{%
\begin{theorem}\label{thm:minimization-hard-to-approximate}
  The existence of an algorithm approximating \mapmrf{} on planar graphs
  with maximum degree 4 and nonpositive $\phi_i$ and $\psi_{ij}$
  to any multiplicative factor implies P = NP\@.
\end{theorem}
}

\later{%
\begin{proof}
Proof of this theorem is a modification of a proof
  of a weaker theorem by Wang~\cite{Wang}.

Given a planar graph $G$,
  we construct an \mapmrf{} instance which has a score of 0
  if and only if $G$ is 3-colorable. 
As planar 3-colorability is NP-complete even on planar graphs of
  maximum degree 4~\cite{color},
  and an approximation algorithm to a multiplicative factor must find
  a solution of weight 0 if one exists,
  this implies the theorem.

The \mapmrf{} instance operates on $G$ with $L=3$ and
  functions
\begin{align*}
  \phi_i(x) &= 0 && \forall\; x \in \{1,2,3\}, i\in V,\\
  \psi_{i,j}(x,y) &= \begin{cases} 0 &\text{if } x \neq y\\
  -1 &\text{if } x = y\\
  \end{cases}
  && \forall\; (i,j)\in E
\end{align*}

An assignment of score 0 is a 3-coloring where the labels are colors;
  the coloring is proper, as any edge with both endpoints of the same color
  would imply the value of the \mapmrf{} instance is negative.
Similarly, a 3-coloring induces an assignment of score 0.
\end{proof}
}

In many applications, MRF is used to minimize an energy function.
Notice that this is equivalent to maximizing the negative energy function. 
Thus Theorem~\ref{thm:minimization-hard-to-approximate} implies
  minimization is inapproximable to any multiplicative factor,
  even if the energy function is nonnegative.

A \emph{polynomial-time approximation scheme} (PTAS) is an algorithm that, given
  an instance of a maximization (minimization) problem and a
  precision parameter $0 < \ep < 1$,
  returns a $(1-\ep)$-approximate ($(1+\ep)$-approximate, resp.) solution
  in time polynomial
  in the size of the instance (with a possible exponential dependence
  on $1/\ep$).
An \emph{efficient} PTAS (EPTAS) is one with runtime of
  the form~$O(f(\ep) \poly(n))$,
  where~$n$ is the size of the instance and $f$ is a
  computable function.

Our main result is:
\begin{theorem}\label{thm:main}
  There is a PTAS for \mapmrf{} in planar graphs
    when all $\phi$ and $\psi$ are nonnegative functions.
\end{theorem}

We also consider the closely related \textsc{Correlation Clustering} problem.
In this, one is given a graph and tasked with partitioning the vertices
  into an arbitrary number of clusters.
The edges have associated rewards and preferences as to
  whether their endpoints should or should not be
  in the same cluster; the objective function is the
  sum of the weights of the edges whose preferences are satisfied.
\textsc{Correlation Clustering} is sometimes expressed with a penalty for
  unsatisfied edges in addition to, or instead of,
  a reward for satisfied edges.
These formulations all have the same optimal solution, but as in \mapmrf{},
  the value of the objective function changes,
  and thus approximation results may differ as well.

Formally, the version we will address is:
\begin{mdframed}[nobreak=true]
  \noindent \textsc{Correlation Clustering}

  \begin{description}
    \item[Instance:]\ 
      \begin{itemize}
        \item graph $G = (V, E)$,
        \item edge preferences $p : E \to \{0, 1\}$,
        \item edge reward function $w : E \to \mathbb{R}_{\geq 0}$.
      \end{itemize}

    \item[Solution:] a partition of the vertices into clusters.

    \item[Maximize:]
      \[
        \sum_{(u,v) \in E} w(u,v)\left[ (1-p(u,v))C(u,v) + p(u,v)(1-C(u,v)) \right]
      \]
      where $C(u,v)$ is 1 if $u$ and $v$ belong
      to the same cluster and 0 otherwise.
  \end{description}

\end{mdframed}

Via a simple reduction to \mapmrf{}: 
\both{%
\begin{corollary}\label{cor:cc}
  There is an EPTAS for \textsc{Correlation Clustering} in planar graphs.
\end{corollary}
}

\later{%
\begin{proof}
We present an approximation-preserving reduction
  from \textsc{Correlation Clustering} to \mapmrf; with that,
  Theorem~\ref{thm:main} gives the result.

Given an instance $\langle G,w,p\rangle$ of \textsc{Correlation Clustering}
    where $G$ is planar,
    we construct an instance of \mapmrf{} with the same graph,
    $L=4$, $\phi_v(x_v)=0$ for all $v\in V, x_v\in \{1,2,3,4\}$, and
    \[
      \psi_{uv}(x_u, x_v) = \begin{cases}
        w(u,v) & \text{if } p(u,v) = 0 \text{\ and } x_u = x_v \\
        0      & \text{if } p(u,v) = 0 \text{\ and } x_u \neq x_v \\
        0      & \text{if } p(u,v) = 1 \text{\ and } x_u = x_v \\
        w(u,v) & \text{if } p(u,v) = 1 \text{\ and } x_u \neq x_v
      \end{cases}.
    \]
    
If $\bm{x}$ is an assignment to this \mapmrf{} instance,
  we make a cluster out of each maximal connected subgraph
  with the same label.

Edges with endpoints of different labels are exactly
  the edges between clusters, so the value of this \textsc{Correlation Clustering}
  solution is the same as the value of $\bm{x}$.

In the other direction,
  we contract each cluster of a given partition down into a single supervertex to
  yield graph $G'$, which is also planar.
By the 4-color theorem,
  there exists an assignment of the labels $\{1,2,3,4\}$
  to the vertices of $G'$ such that no adjacent vertices have the same label. 

Give each vertex in $G$ the same label as the corresponding supervertex in $G'$.
Edges within a cluster have both edges corresponding to the same supervertex,
  and thus they have the same label.
Edges between clusters have corresponding edges in $G'$,
  and thus have endpoints with different labels.
Thus the value of the assignment is exactly the value of the partition.

Both the creation of the corresponding \mapmrf{} instance and the
  conversion of a solution of that instance to
  a solution of \textsc{Correlation Clustering} take time linear
  in the size of the input.
Thus there is a linear time approximation-preserving reduction,
  which, in conjunction with Theorem~\ref{thm:main} completes the proof.
Note that while the PTAS for \mapmrf{} is \emph{not} an \emph{efficient} PTAS,
  this one is, because $L=4=O(1)$.
\end{proof}
}

\subsection{Outline}

In Section~\ref{sec:prior},
  we review past work on~\mapmrf{}.
Next, in Section~\ref{sec:bw}, we give an exact algorithm for
  graphs of bounded branchwidth.
Then, Section~\ref{sec:ptas} proves Theorem~\ref{thm:main}.
\ifabstract%
In the interest of space, proofs of Corollary~\ref{cor:cc} and Theorems~\ref{thm:general-hard-to-approximate} 
and~\ref{thm:minimization-hard-to-approximate} are in the appendix.
\fi
We demonstrate some promising experimental results in Section~\ref{sec:experiments} with applications to computer vision.
Finally, we offer discussion in Section~\ref{sec:conclusions}.

\section{Prior Work}\label{sec:prior}

Markov Random Fields originated in statistical physics as a
  generalization of the Ising Model~\cite{MRFapps}.
There are numerous techniques to solve \mapmrf{}, both in general and
  on specific instances; some are outlined here.

An MRF is \emph{binary} if there are exactly two labels and \emph{submodular}
  if for all~$u,v \in V$,~for all $i, j \in \{1,\ldots, L\}$,
  $\psi_{u,v}(i,i)+\psi_{u,v}(j,j) \geq \psi_{u,v}(i,j)+\psi_{u,v}(j,i)$.

If an MRF is both binary and submodular, \mapmrf{} can be solved exactly
in polynomial time by reduction to \textsc{Min-Cut}~\cite{boy}.
If the graph is also planar,
  the running time can be improved to $O(n \log(n))$~\cite{schmidt}.

For MRF in graphs which have bounded degree and an excluded
  minor (which includes all bounded degree planar graphs)
  Jung and Shah use techniques similar to ours to find a
  PTAS with running time doubly exponential in $1/\ep$~\cite{Jung}.
For the alternate formulation of \textsc{Correlation Clustering} which seeks to minimize penalties for unsatisfied edges, Klein et al.\ demonstrate a non-efficient PTAS in planar graphs~\cite{CC}.

When $\psi_{ij}$ are defined by a metric on the labels,
  the problem is referred to as
  \textsc{Metric Labeling};~\cite{KT02} provides a
  $O(\log L \log\log L)$-approximation algorithm
  for the problem.

The \textsc{Generalized Potts Model}, from statistical mechanics, is a restriction
of MRF that reduces to the classic \textsc{Multiway Cut} problem;~\cite{BVZ98} uses
local search to approximate this model.
\textsc{Multiway Cut} is a special case
  of \textsc{Metric Labeling}, where some vertices are forced
  to have particular labels.
In planar graphs, there is a PTAS for the problem~\cite{KleinMTC}.
In general, there are constant-factor approximations~\cite{DahlhausJPSY92}.

\textsc{0-Extension} is a generalization of \textsc{Multiway Cut}
in which the cost of the edge depends on the specific terminals
associated with the edge's endpoints, not just whether the terminals are the same.
In general graphs, this problem
  is~$O(\log L / \log \log L)$-approximable~\cite{Karzanov98,CKR01b,FHRT03}
  and can be approximated to a constant-factor in planar graphs.

Various heuristics exist to approximate MAP on planar graphs and are used extensively
  in computer vision for applications such as:
\begin{itemize}
  \item Stereo vision: given two photographs taken side-by-side,
          estimate the depths of each pixel.
  \item Object segmentation: find the boundaries of objects in photographs.
  \item De-noising: remove grainy noise from an image.
  \item Photomontage: combine several images into one.
\end{itemize}
Two standard benchmarks for these problems are OpenGM~\cite{opengm} 
and the Middlebury stereo dataset~\cite{middlebury}.
For a detailed treatment of MRF as applied to computer vision,
  see, e.g.,~\cite{yarkony}.

Many problems, including \mapmrf{} and more traditional optimization
  problems such as \textsc{TSP}, \textsc{Steiner Tree}, \textsc{Vertex Cover}, 
  \textsc{Graph Coloring}, \textsc{Clique}, \textsc{Hamiltonian path},
  and \textsc{Feedback Vertex Set} can be solved exactly in polynomial time
  on graphs of bounded \emph{branchwidth}. 
Branchwidth, like treewidth, pathwidth, bandwidth, outerplanarity,
  or cliquewidth, is a measure
  of the ``simplicity'' of a graph.
These measures are amenable to dynamic programming and have been
  of great importance when designing approximation schemes
  on planar graphs~\cite{baker,bidim,KleinTSP,KleinSteiner}.

Our algorithm draws inspiration from Baker's technique~\cite{baker},
  a powerful framework for building PTASes in planar graphs.
In a nutshell, Baker guesses a way to decompose a graph
  into a number of smaller graphs of bounded outerplanarity.
These smaller graphs are each solved optimally and independently,
  and then combining the solutions incurs at most $\ep \OPT$ error.
This technique was originally applied to \textsc{Independent Set} but can be
  used for a number of problems, such as \textsc{Vertex Cover}, 
  \textsc{Edge-Disjoint Triangles}, and
  \textsc{Dominating Set}~\cite{baker}.

Recently, Wang posted a manuscript on
  arXiv claiming a PTAS for \mapmrf{}
  on planar graphs, among other results~\cite{Wang}.
We remark that our main result, Theorem~\ref{thm:main},
  was discovered independently.
Theorem~\ref{thm:minimization-hard-to-approximate} draws inspiration from
  and strengthens a hardness proof of Wang.
Unfortunately, there appears to be a bug in an
  vital lemma in~\cite{Wang}.
\ifabstract%
We discuss this in Appendix~\ref{sec:wang}.
\fi

\later{%
\ifabstract%
\section{Discussion of~\cite{Wang}}\label{sec:wang}
\fi
Lemma~4.2 of~\cite{Wang} is critical to the correctness of Wang's PTAS\@;
  as presented, it has some problems.

The stated runtime does not account for the
  degree of the graph;
  $f_i$ has $L^{d+1}$ possible outputs if vertex $v_i$ has degree $d$; 
    all possible outputs must be examined to ensure correctness.

Additionally, $S^U_{i \setminus p_i}$ is defined to be the max-sum of
  the liberal functions attached to vertices
  of $(U \cap V_{T_i}) \setminus (X_{p_i} \cup \delta X_{p_i})$.
In a nice tree decomposition of a star, that resulting set is empty for
  all $i$ except the root $r$, which means that
  the entire value of $S^U_{r \setminus p_r}$ is $\Gamma^{\sigma_{i \setminus p_i}}_{X_r \setminus X_{p_r}}$.
$\Gamma^{\sigma_{i \setminus p_i}}_{X_r \setminus X_{p_r}}$,
  in this case, is defined to be the sum of liberal functions attached
  to every vertex in the star when the configuration of \emph{just} the root
  is fixed to be $\sigma_{i \setminus p_i}$.
So, calculating $\Gamma^{\sigma_{i \setminus p_i}}_{X_r \setminus X_{p_r}}$
  is equivalent to solving the original problem
  and how it is calculated is not specified.
}

\section{\mapmrf{} in Bounded Branchwidth Graphs}\label{sec:bw}

A branch decomposition of a graph~$G=(V,E)$ is an unrooted binary
  tree~$T$ whose leaves are the edges $E$ of $G$. 
Deleting an edge of~$T$ generates two subgraphs of~$G$,
  each induced by the edges in one component of $T$.
Some vertices are contained in both subgraphs. 
The maximum number of these overlapping vertices for any such pair
  of subgraphs is the \emph{width} of the decomposition. 
The minimum width of any branch decomposition of $G$ is its \emph{branchwidth}.

Our PTAS is an application of Baker's technique~\cite{baker},
  and works by breaking up the problem into bounded branchwidth subproblems,
  each of which can be solved exactly in polynomial time. 

\begin{theorem}\label{thm:bw}
Given an \mapmrf{} instance $(G = (V, E), L, \phi, \psi)$ and a
branch decomposition $T$ of width $k$,
an optimal solution can be found in time $O(|E|kL^{2k})$.
\end{theorem}
\begin{proof}
We use dynamic programming.
$T$ will guide the dynamic program, and thus we want a root with two children.
To that end, we choose an arbitrary edge of~$T$
  and subdivide it with a new vertex~$r$ that we designate the root. 
Now~$T$ is a rooted binary tree
  but maintains the other properties of a branch decomposition.

With each tree vertex $v \in T$,
  let~$G(v)$ be the subgraph of $G$ induced
  by the edges of~$G$ which are descendants of~$v$. 
Observe that~$G(r) = G$. 
Denote by~$\delta(G(v))$ the vertices of~$G(v)$ which are incident to
  edges not in $G(v)$.
Note that~$|\delta(G(v))| \leq k$ for all~$v \in T$.

For each vertex $v \in T$,
  we will compute the assignment to the vertices~$V(G(v)) - \delta(G(v))$
  for each possible assignment to the vertices~$\delta(G(v))$
  which maximizes the score of the MRF on~$G(v)$.
This is done bottom-up,
  so that for all non-leaf vertices of~$T$,
  assignments for both of their children are computed first.

If $v$ is a leaf, $G(v)$ is a single edge with its endpoints. 
Thus either $V(G(v)) - \delta(G(v))$ is empty and
  finding the optimal assignment is trivial;
  or~$V(G(v)) - \delta(G(v))$ is a single endpoint,
  and all possible label assignments can be tested.
In both cases, it takes~$O(L^2)$ time to test
  for all possible boundary assignments
  what the best assignment to~$V(G(v)) - \delta(G(v))$ is.

If $v$ is not a leaf, it has two children~$u_1, u_2$.
Let~$U = \delta(G(u_1))\cup\delta(G(u_2))$
  and~$I = \delta(G(u_1)) \cap \delta(G(u_2))$.
Notice $\delta(G(v)) \subseteq U$.
For each label assignment to the vertices of $\delta(G(v))$,
  the best assignment to $V(G(v)) - \delta(G(v))$ is the union of best
  assignments to $V(G(u_1)) - \delta(G(u_1))$ and $V(G(u_1)) - \delta(G(u_1))$
  for some assignment to $I - \delta(G(v))$,
  and its value is the sum of the values of those assignments
  minus the values of $\phi$ on $I$.
As those assignments and values have already been computed,
  finding the optimal ones can be done in time 
  $O(|I|L^{|U|})$.
Since~$|I| \leq k$ and~$|U| \leq 2k$,
  computing all the assignments and values at vertex~$v$
  takes time~$O(k^{2k})$.

$\delta(G(r))$ is empty, so the unique assignment and value computed at $r$ are the exact optimal solution to the \mapmrf{} instance. 
The rooted branch decomposition has $2|E|-1$ vertices,
  thus the running time is $O(|E|kL^{2k})$.

\end{proof}

We summarize the algorithm:
\begin{mdframed}[nobreak=true]
\begin{enumerate}
  \item Choose an arbitrary edge $e$ of $T$, and subdivide it with a new root vertex~$r$.
  \item With each vertex $v$ of $T$ associate the subgraph $G(v)$ of $G$ induced by the edges of $G$ which are descendants of $v$ (with respect to the root $r$).
  \item Consider each vertex $v$ of $T$ from leaf to root:
  \begin{enumerate}
    \item If $v$ is a leaf, for each possible label assignment to the vertices of $\delta(G(v))$, by brute force, compute the best assignment to $V(G(v)) \setminus \delta(G(v))$.
    \item Otherwise, for each possible label assignment to the vertices of $\delta(G(v))$, combine the values and assignments of $v$'s two children
        to determine the best assignment to $V(G(v)) \setminus \delta(G(v))$.
  \end{enumerate}
  \item Return the best assignment for $G(r) = G$.
\end{enumerate}
\end{mdframed}

\section{PTAS for \mapmrf{} on Planar Graphs}\label{sec:ptas}
We now give the PTAS for our main result.
As input, we are given an instance of \mapmrf{} $\langle G=(V,E), L, \phi, \psi \rangle$
  where $G$ is a planar graph,
  and a desired error parameter $0 < \ep < 1$,
  with $k = \frac{1}{\ep}$.

Fix some vertex $r$.
We say an edge has $r$-level $d$
    if one of its endpoints is hop-distance
    $d-1$ from $r$ and the other is hop-distance $d$.
Let $G_j$ be the graph resulting in deleting all edges 
  with $r$-levels congruent to $j \pmod{k}$.

The algorithm is:
\begin{mdframed}[nobreak=true]
\begin{enumerate}
  \item Choose a vertex $r$ arbitrarily.
  \item Let $k = \frac{1}{\ep}$.
  \item For each $j \in \{0,\ldots,k-1\}$:
    \begin{enumerate}
      \item Compute $G_j$.
      \item Find an approximate branch decomposition $T$ of each
            component of $G_j$ using the algorithm in~\cite{Tamaki}.
      \item Apply Theorem~\ref{thm:bw} to each component of $G_j$ and combine
            the resulting best label assignments into $\bm{x}_j$.
      \item Compute the value $h_j$ of the objective function on $G$
              from $\bm{x}_j$.
    \end{enumerate}
  \item Return the assignment corresponding to the largest $h_j$.
\end{enumerate}
\end{mdframed}

With this, we are ready to prove our main result.

\begin{proof}[of Theorem~\ref{thm:main}]
First, we tackle the runtime.
For each $j$, it takes linear time to construct $G_j$ by building a breadth first search tree from $r$.

By construction, there exists in each component of $G_j$ a path of length at most $k$ from each vertex to a vertex on the face containing $r$. An algorithm by Tamaki~\cite{Tamaki} allows us to construct a branch decomposition of width at most $2k$ on a graph with this property in time $O(m_i2^{2k})$, where $m_i$ is the number of edges in the component.

Then, solving these optimally using~\ref{thm:bw} and combining takes time $O(|E|kL^{4k})$.
As we try $k$ different choices of $j$,
 the total running time is $O(|E|k^2L^{4k})$.
This is linear in the size of the graph, as $k$ is a function of $\ep$.
However, as $L$ is part of the input, this is not an efficient PTAS\@.

Now, we demonstrate correctness.
Let $\bm{x}^*$ be an optimal label assignment.
By construction, $\bm{x}_j$ is the optimal assignment on $G_j$.
Let $H_j$ be the objective function restricted to $G_j$.
Since~$\bm{x}_j$ consists of optimal solutions of each component of~$G_j$,~$H_j(\bm{x}_j) \geq H_j(\bm{x}^*) \label{eq1}$.

Let~$d_j = H(\bm{x}^*) - H_j(\bm{x}^*)$.
So we have~$H(\bm{x}_j) \geq H(\bm{x}^*) - d_j$.
Summing over all choices of~$j$,
\[
\sum_{j=0}^{k-1} H(\bm{x}_j) \geq \sum_{j=0}^{k-1} H(\bm{x}^*) - d_j(\bm{x}^*).
\]

Each edge in $G$ is missing from at most one $G_j$, so 
$\sum_{j=0}^{k-1} d_j \leq H(\bm{x}^*)$.
Thus,
\begin{align*}
\sum\limits_{j=0}^{k-1} H(\bm{x}_j) &\geq k H(\bm{x}^*) - H(\bm{x}^*) 
  = k(1-1/k)H(\bm{x}^*) 
  = k(1-\ep)H(\bm{x}^*).
\end{align*}

Consequently, there exists some $j$ where
$H(\bm{x}_j) \geq (1-\ep)H(\bm{x}^*)$.
\end{proof}

\section{Experiments}\label{sec:experiments}

The approximation scheme has relatively small constants,
  which suggested that it might be feasible to use in practice.
We implemented a version of this PTAS in \verb|C++11|
  for tasks that arise in computer vision.
For simplicity, we restricted our implementation to grid graphs, as
  is common in image processing.
Optimal branch decompositions are particularly easy to find in this domain.

\subsection{Stereo Matching}

Given two images representing a left camera angle and a right camera angle
  and a number $L$ of relative depth labels, we wish to assign a label
  in $\{1, \ldots , L\}$ to each pixel in the, say, left image.
In the computer vision community, these are often visualized as
  \emph{disparity maps}, or grayscale images of the relative depths;
  see e.g.~Figure~\ref{fig:artifacts}.
We use the 16 label \emph{tsukuba} example from the Middlebury stereo benchmark~\cite{middlebury} for illustration here:

\begin{figure*}[h]
\centering
\mbox{\begin{subfigure}[b]{0.5\textwidth}
\centering
\includegraphics[width=6cm]{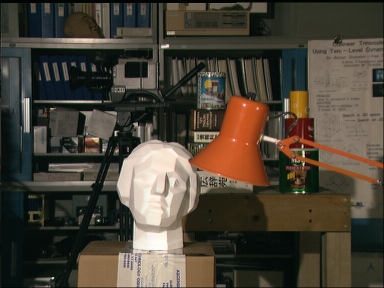}
\caption{Left Image}
\end{subfigure}

\begin{subfigure}[b]{0.5\textwidth}
\centering
\includegraphics[width=6cm]{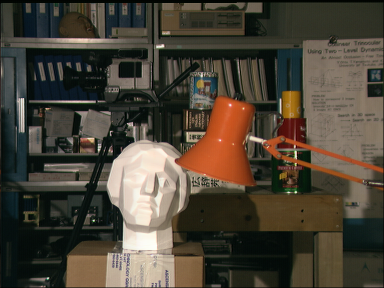}
\caption{Right image}
\end{subfigure}
}
\caption{\emph{Tsukuba} images from the Middlebury stereo benchmark.}
\end{figure*}

We used the following model as input to our algorithm.
The graph $G = (V, E)$ is the planar grid graph
  where each vertex represents a pixel.
We define functions
\begin{align*}
  \phi_{u}(i) &= \beta-\|u - u^{(-i)}\|_2^2 && \forall \; u \in V \\
  \psi_{u,v}(i,j) &= \begin{cases} 0 & \text{if $i=j$} \\
      \beta-\|u - v\|_2^2 & \text{if } i\neq j \end{cases} && \forall \; (u, v) \in E
\end{align*}
where $u$ is a pixel in the left image, $u^{(-i)}$ is the pixel
  that is $i$ columns to the left of the pixel corresponding to $u$ 
  in the right image, $\|\cdot \|_2^2$ is square 2-norm in CIELUV color space,
  and $\beta$ is a constant sufficiently large to ensure that all outputs of
    the functions are positive. 

In addition to our basic algorithm, we also incorporate a few very simple
  vision-specific heuristics to refine our results.
Initializing boundary pixels to the values
  from the previous (either left or right) connected component
  yields more visually continuous results.
Since the analysis of the approximation holds for any
  value of the boundary pixels, in particular it holds for these values.
Thus the approximation guarantee is preserved at this step.
However, this results in some visual artifacts (see Figure~\ref{fig:artifacts}).

\begin{figure*}[h]
\centering
\mbox{\begin{subfigure}[b]{0.5\textwidth}
\centering
\includegraphics[width=6cm]{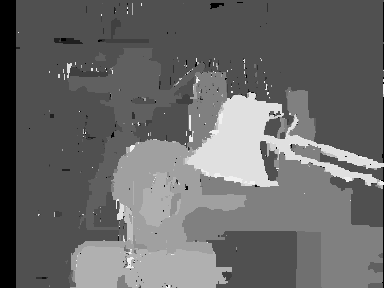}
\end{subfigure}

\begin{subfigure}[b]{0.5\textwidth}
\centering
\includegraphics[width=6cm]{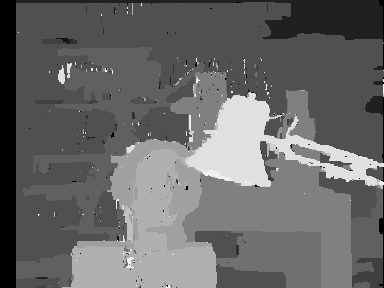}
\end{subfigure}
}
\caption{\label{fig:artifacts} Visual artifacts after one heuristic.}
\end{figure*}

To remedy this, we run the algorithm twice
  (intuitively, left-to-right and then right-to-left)
  and combine the solutions in an approximation-preserving way.

Finally, a tiny amount of smoothing is done to remove remaining noise;
  this does not guarantee the approximation but leads to more visually-pleasing
  results.

\later{%
\newpage
\section{Additional Figures}
\begin{figure*}[h]
\centering
\mbox{\begin{subfigure}[b]{0.5\textwidth}
\centering
\includegraphics[width=6cm]{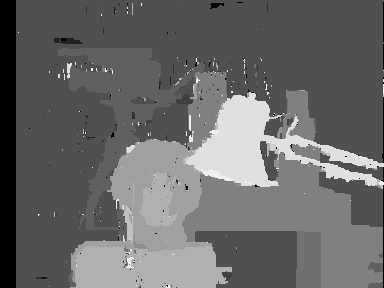}
\caption{With passes combined}
\end{subfigure}

\begin{subfigure}[b]{0.5\textwidth}
\centering
\includegraphics[width=6cm]{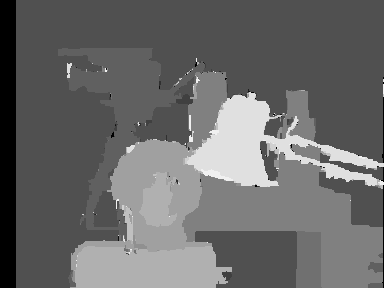}
\caption{with smoothing}
\end{subfigure}
}
\mbox{\begin{subfigure}[b]{0.5\textwidth}
\centering
\includegraphics[width=6cm]{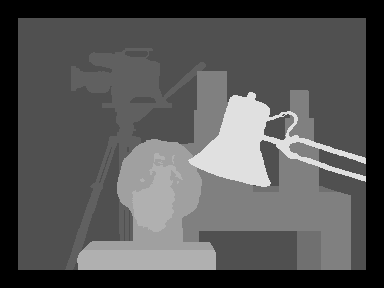}
\caption{Ground truth for comparison}
\end{subfigure}

\begin{subfigure}[b]{0.5\textwidth}
\centering
\includegraphics[width=6cm]{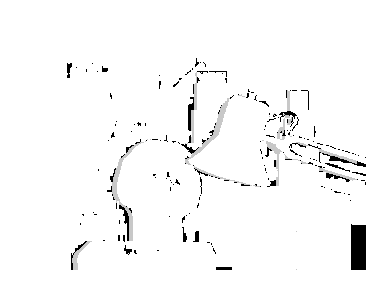}
\caption{\label{fig:mislabeled} Mislabeled pixels highlighted}
\end{subfigure}
}
\caption{\label{fig:more-heuristics} Our results on \emph{tsukuba} with heuristics applied.}
\end{figure*}
}

We use the evaluation tools provided on the Middlebury stereo website: 
  5.07\% of all pixels are mislabeled including 3.02\% of non-occluded regions
  and 11.5\% of regions near depth discontinuities.
Furthermore, as seen in Figure~\ref{fig:mislabeled},
  a large fraction of mislabeled pixels are concentrated in the bottom right;
  we believe that discrepancies between the MRF model and the ground truth explain
  this.

State-of-the-art algorithms mislabel a little more than 1\% of pixels
  including typically over 4\% of regions near depth discontinuities.
Many of the published algorithms on the Middlebury benchmark mislabel
  significantly more than 5.07\% of all pixels,
  and the best algorithms involve optimizing dozens of hyperparameters
  and are highly specialized to their applications.

We found that our generic PTAS required only a few basic heuristics to perform
  quite well, suggesting that with a few more heuristics, it could be very
  competitive.

\subsection{Observed $\ep$ dependencies}

Experiments support the theoretical dependencies on $\ep$.
Figures~\ref{fig:score} and~\ref{fig:runtime} show the
  score and log running time, respectively,
  of our algorithm on the \emph{tsukuba} image as a function of $1/\ep$,
  using 14 labels and the learned parameters.
The score changes remarkably little, considering the improvement
  in the theoretical bound.

\begin{figure*}[h]
\centering
\mbox{\begin{subfigure}[t]{0.5\textwidth}
\includegraphics[width=\textwidth]{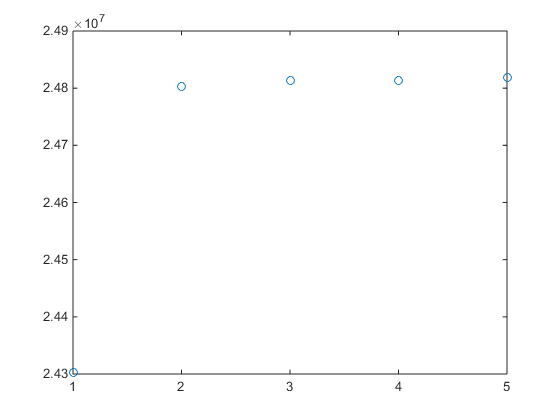}
\caption{\label{fig:score} Score (in arbitrary units) as a function of $1/\ep$.}
\end{subfigure}
\begin{subfigure}[t]{0.5\textwidth}
\includegraphics[width=\textwidth]{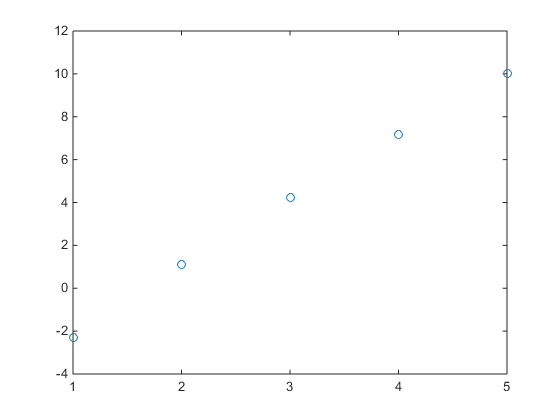}
\caption{\label{fig:runtime} Log of running time (in seconds) as a function of $1/\ep$. Experiments were run on a mid-range 2014 laptop.}
\end{subfigure}
}
\caption{}
\end{figure*}

The running time matches the theory very closely.
The observed ratios of running time as $\ep$ increases
  from 2 to 5 are 23, 18.45, and 17.8;
  the theoretical run time is proportional
  to $1/\ep \cdot L^{1/\ep}$ which would predict ratios of 21, 18.67, and 17.5.

\subsection{OpenGM benchmark}
We used the OpenGM 2.3.3~\cite{opengm} library to benchmark the
  actual energy minimization performance of our algorithm compared to other existing methods.
Our algorithm was run with $\ep = 1/3$.

On the Inpainting benchmark, our algorithm achieves a score of 461.82,
  which is about 1.6\% away from the best algorithm's and
  better than half of the competing algorithms.
On the Object Segmentation benchmark, we perform a bit worse;
  our score is about 64\% away from the best and
  worse than most of the competition.

\section{Discussion \& Conclusions}\label{sec:conclusions}

Our algorithm gives the first known PTAS for maximum a posteriori assignment
  on \mapmrf{}, and the first EPTAS for this variant of
  \textsc{Correlation Clustering} in planar graphs.
Combined with our hardness results,
  much of the complexity of \mapmrf{} on planar graphs is now settled.
While the algorithm is not directly competitive with the state of the art for
  computer vision tasks, it is sufficiently close to those algorithms
  to suggest applications in improving them,
  as well as in other applications which lack specialized algorithms. 

One can readily extend the given PTAS to
  more general classes of graphs, or (non-pairwise) MRFs in planar graphs
  with bounded factor degree.

Compelling future research directions
  include studying \mapmrf{} with
  with negative functions
    and two labels (but not necessarily submodular),
    and with more than two labels but submodular functions.

\bibliographystyle{plain}
\bibliography{Paper}

\ifabstract%
\newpage
\appendix
\section{Elided Proofs}
\magicappendix%
\fi
\end{document}